\documentclass{article}
\usepackage{graphicx}
\usepackage{pstricks}
\usepackage{authblk}
\usepackage{url}
\usepackage{amssymb}
\usepackage{needspace}
\usepackage{amsthm,amsmath,amssymb,amsfonts, fullpage}
\usepackage{dsfont}
\newcommand{\textsfup}[1]{\textsf{\textup{#1}}}

\newcommand{\ket}[1]{| #1 \rangle}

\newcommand{\comment}[1]{}
\newcommand{\hide}[1]{}

\newcommand{\eps}{\varepsilon}
\renewcommand{\phi}{\varphi}

\newcommand{\norm}[1]{\parallel #1 \parallel}

\newcommand{\MITM}{\text{Meet-in-the-middle}}
\newcommand{\GCF}{\Gamma_{CF}}
\newcommand{\GKE}{\Gamma_{KE_2}}

\newcommand{\SSS}{\mbox{\textsfup{S}}}
\newcommand{\UU}{\mbox{\textsfup{U}}}
\newcommand{\CC}{\mbox{\textsfup{C}}}

\newcommand{\search}{\textsfup{SEARCH}}
\newcommand{\ED}{\text{Element Distinctness}}
\newcommand{\advpm}{\textsf{ADV}^{\pm}}

\newtheorem{theorem}{Theorem}
\newtheorem{corollary}{Corollary}
\newtheorem{definition}{Definition}

\newtheorem{lemma}{Lemma}

\title{Quantum attacks against iterated block ciphers}
\author{M. Kaplan\thanks{LTCI, T\'el\'ecom ParisTech, 23 avenue dItalie, 75214 Paris CEDEX 13, France}}
\date{}
\begin{document}

\maketitle

\begin{abstract}
We study the amplification of security against quantum attacks provided by iteration of block ciphers.
In the classical case, the $\MITM$ attack is a generic attack against those constructions.
This attack reduces the time required to break double iterations
to only twice the time it takes to attack a single block cipher, given that the attacker has access to a large
amount of memory. More abstractly, it shows that security by composition does
not achieve exact multiplicative amplification.

We prove here that for quantum adversaries, two iterated ideal block cipher are more difficult to attack 
than a single one. We give a quantized version of the $\MITM$ attack and 
then use the generalized adversary method to prove that it is optimal.
An interesting corollary is
that the time-space product for quantum attacks is very different from  what classical attacks allow.
This first result seems to indicate that composition resists better to quantum attacks
than to classical ones.

We investigate security amplification by composition further by examining the case 
of four iterations. 
We quantize a recent technique called the dissection attack using the framework of quantum walks.
Surprisingly, this leads to better gains over classical attacks than for double iterations,
which seems to indicate that when the number of iterations grows, 
the resistance against quantum attacks decreases.

\end{abstract}

\section{Introduction}
Quantum information processing has deeply changed the landscape of classical cryptography. In particular, cryptosystems based on integer factoring and discrete logarithm are known to be completely insecure against quantum computers. This opened the field of {\em post-quantum (PQ) cryptography}, which tries to restore the security of classical cryptosystems against quantum attacks. 
Regarding public key cryptography, PQ cryptography is investigating (or reinvestigating) a number of different approaches such as lattice-based cryptography, code-based cryptography or elliptic curve cryptography. The goal of these different approaches is to base
public key encryption mechanisms on 
problems that are believed to be hard to solve even for quantum computers.

The situation in symmetric cryptography is more complicated.
It is well known quantum computers can speed-up certain information processing tasks 
that are useful for the cryptanalysis of symmetric cryptosystems.
This includes
exhaustive search~\cite{gro96} or collision finding~\cite{amb07}.
In fact, quantum cryptanalysis is often cited as a 
motivation to study quantum algorithms~\cite{BHT98,AS04,Zhan13}.
Since these quantum algorithms usually allow only polynomial speedups, this tends to spread 
the belief that the security against quantum adversaries can be restored by increasing the size of private keys.

This short answer to the question of security against quantum attacks leaves aside a number of issues.
For example, it applies to generic attacks and not to cryptographic attacks. Roughly speaking,
generic attacks work against constructions based on ideal functionalities, whereas cryptographic 
attacks try to attack their implementations. 
Attacking realistic, complex cryptosystems may require more effort than just applying basic quantum
algorithms. This situation is well described by Daniel Bernstein~\cite{ber10}:
\begin{quote}
``Grover’s algorithm takes only square-root time compared to a brute-force key
search, but this does not mean that it takes less time than the more sophisticated
algorithms used to find hash collisions, McEliece error vectors, etc. Sometimes
Grover's idea can be used to speed up the more sophisticated algorithms, but
understanding the extent of the speedup requires careful analysis."
\end{quote}
In this work, we argue that the tools developed to study quantum speedups in complexity
theory can be applied to get new insights in cryptographic settings.
In particular, the field of quantum walk algorithms is a flexible framework to develop search algorithms on complex structures~\cite{santha}.
Even more interestingly, the generalized adversary method~\cite{hls07} can be used to prove lower bounds, and thus resistance against quantum attacks, at least in a black box setting. 
Applying these tools to
study the resistance against quantum attacks
is an extension of the current horizon of PQ cryptography that can be though as
{\em quantum post-quantum cryptography}.

We focus here on one of the most fundamental situations in symmetric cryptography: block cipher encryption. 
In this setting, a message $m$ is cut into blocks of fixed size, and each block is encrypted using a permutation specified
by a secret key. The original message can then be recovered by applying the inverse permutation on each
block. The specification of the encryption and decryption algorithms are public, only the key remains secret through the process.
Block cipher encryption is widely used in practice. It is also an important building block in the design
of other cryptographic primitives such as message authentication codes, hash functions, or digital signature.

A block cipher is designed such that applying a permutation specified by an randomly sampled
key ``looks like" sampling a random permutation.
In this paper, we work at a more abstract level and consider that a block cipher is a collection of
random permutations $F_i:[M] \rightarrow [M]$ where $i=1,\dots, N$ are the potential secret keys.
The set of permutations is public, and anyone can efficiently compute $F_i(X)$ and $F_i^{-1}(X)$ for any $X,i$.
We consider an attacker that knows a few pairs of plaintext with corresponding ciphertexts, all encrypted with the same key. The goal for the attacker is to extract the secret key that was used for encryption from the datas.

If the permutations are random, the only attack is to search exhaustively among the keys,
which requires $O(N)$ classical queries to the functions $F_i$ or the inverse functions $F_i^{-1}$. Of course, a quantum attacker
can use Grover search algorithm to achieve a quadratic speedup, and find the key with $O(\sqrt N)$ quantum queries.
In order to restore the same security parameter
against a quantum cryptanalyst, it suffices to double the length of the key. 
Of course, it is unreasonable to assume that a time bound $N$ is
considered secure against quantum attacks because it is secure against classical attacks.
However, the broader question we investigate here is
to amplify the security against quantum attacks.
We believe this question remains legitimate.

Although increasing the key length is a neat theoretical answer, it is not always clear how to implement it in practice.
Block ciphers are deterministic algorithms, designed to work with specified parameters $N$ and~$M$. Assuming that one can double the 
key size is a strong structural assumption on the design of the cipher. For example, AES can be used to work with different key sizes, but it was not the case
for its ancestor DES. This block cipher was designed to work with 56 bits keys, and although this was considered sufficient by the time is was standardized, the question of how to increase this size became central when brute-force attacks started looking
realistic.

The simplest attempt to increase the key size is to compose permutations with independent keys.
The composition of $r$ independent permutations of a block cipher is called an $r$-encryption.
For $r=2$, the size is doubled, but there is a clever attack against 
this construction~\cite{DH77,MH81}. 
Suppose that an attacker knows a pair of plaintext-ciphertext $(P,C)$. These satisfy $C=F_{k_2}(F_{k_1}(P))$,
where $(k_1, k_2)$ are the keys used for encryption.
Since inverse permutations can be computed, an attacker can construct tables $F_k(P)$ and 
$F^{-1}_{k'}(C)$ for every possible keys $k, k'$. 
Finding a collision $F_{k_1}(P) = F^{-1}_{k_2}(C)$ reveals the keys used for encryption.

This attack, known as the $\MITM$ attack, 
shows that it only takes twice more time to attack double iterations than it takes to
attack a single one, where 
a naive cryptographer would expect
a quadratic increase. Equivalently, the key expansion obtained by double-iteration is only 1 bit, although the size
of the key space has doubled.
Of course, this attack is optimal up to a factor two.
The $\MITM$ attack shows that even a simple idea such as
security amplification by composition should
be carefully studied.
It also has practical consequences, and led to the standardization
of triple-DES rather than the insecure double-DES.

We address the question of how resistant composition is against quantum adversaries.
An obvious quantum attack against the 2-encryption is to use a collision finding algorithm.
Finding collisions has been studied extensively in classical and quantum settings. The series of quantum algorithms
for this problem culminated with
Ambainis' celebrated algorithm~\cite{amb07} for the Element Distinctness problem,
which is optimal with respect to the number of queries to the input~\cite{AS04}. In our case, it extracts the keys with $N^{2/3}$ quantum queries
to the permutations.
However, the key extraction problem for double encryption has
more structure than element distinctness and there is no clear indication that this approach is optimal.
The problem has a lot more possible inputs and more queries are allowed.
Starting from an instance of element distinctness and embedding it to an input that is consistant with the structure
of key extraction may not always be possible without a large number of queries to the input.
For this reason, there is no obvious way of proving
the optimality of Ambainis' algorithm for key extraction by reduction from Element Distinctness.

In Section~\ref{sec:qmitm},
we prove using the generalized adversary method that $N^{2/3}$ queries are also required to extract the keys in the case of 2-encryption
({\bf Theorem~\ref{thm:KE2LB}}). 
Starting from an adversary matrix for $\ED$, we build an adversary matrix for the new problem ({\bf Lemma~\ref{lm:WCLB}}). The underlying idea is that even if there is no obvious reduction, the two problems are equivalent because they have the same symmetries.
We prove that the attack against 2-encryption that consists in searching for collision is the most time-efficient, 
leading to a generic attack deserving to be called {\em the quantum $\MITM$ attack}.
The immediate consequence is that for quantum computers, contrary to the classical ones, 2-encryption is harder to break than a single ideal cipher.

A surprising corollary is that classical and quantum time-space products are very different ({\bf Corollary~\ref{cor:TS}}). 
In the classical case, the $\MITM$ attack is time-efficient for an attacker that is willing to pay with more space, but the global time-space product is similar to
the one achieved by an exhaustive search. Using quantum algorithms, the time-space product of the 
optimal algorithm of Ambainis
is worse than an exhaustive search. While this may not be a surprise from
the point of view of quantum complexity theory (see e.g. the conclusion of~\cite{ED05}), this suggests that the time-space product,
a common way of evaluating classical attacks~\cite{ddks12}, may not be the correct figure of merit to evaluate quantum attacks.

The results obtained for two iterations suggest that composition could be a good tool to 
amplify the resistance against quantum attacks because it prevents the quadratic speedups
allowed by the quantization of an exhaustive search.
We investigate this question further in Section~\ref{sec:4enc}, by looking at the case of 4-encryption.
Once again, the tools from quantum complexity theory appear to be very helpful to tackle this question.
We give a quantization of the dissection attack for four encryption recently introduced by Dinur, Dunkelman, Keller and Shamir~\cite{ddks12}. This attack is essentially a composition of an exhaustive search with the basic $\MITM$ algorithm.
Surprisingly, it is better than previously known classical attacks solely based on the $\MITM$ technique.
Quantum query complexity theory developed good tools to quantize such compositions, but in order to 
quantify also
time and space complexity, we use the framework of quantum walks. Our main finding in this case is that
the resistance against quantum attacks decreases when the number of iterations goes from two to four ({\bf Theorem~\ref{thm:KE4UB}}).

While these tools appear to be very helpful to study two encryptions and four encryptions,  
we are not in position to make a statement for general multiple encryptions.
Using the generalized adversary method and quantum walks has been very fruitful for studying Merkle puzzles in a quantum world~\cite{BHKKL11}, in which these techniques were used to devise attacks and prove their optimality.
We present here another important cryptographic scenario in which these tools can be applied to derive new results.
A similarity between the two scenarios is that polynomial speedups are very insightful.
Many specifically quantum techniques are available to study such speedups and their optimality in black-box settings.
Even if the main question about successive encryption remains open, we hope that our work demonstrates that quantum techniques can be very interesting for PQ cryptography, and that it
will motivate further interactions between quantum computer scientists and classical cryptographers.

\section{Optimality of the quantum $\MITM$}
\label{sec:qmitm}

In what follows, $M$ and $N$ are two integers of comparable size and $[N]$ is the set of integers from 1 to $N$.
We denote $\mathcal S_{[N]}$ the set of permutations of $[N]$. 
The space of keys is $[N]$ and the space of blocks $[M]$.

We consider problems with inputs given as oracles (or black-boxes). Usually, we consider these inputs as functions $f$, and a classical
query to the input returns $f(x)$ for some $x$ in the domain of $f$. 
In some cases, we may also consider an input $f$ as a string where $f_i$ denotes the result of the classical query $i$ to $f$. 
This notation is convenient in particular when considering adversary matrices whose entries are indexed by inputs to the problem.
In the quantum setting, the only difference is that attacker can make quantum queries to the input.
A brief exposition of the underlying model with
the main theorems that we use to derive our results can be found in Appendix.

The Element Distinctness problem has been extensively studied in 
quantum query complexity.
In particular, Ambainis' quantum walk based algorithm~\cite{amb07} is known to be optimal for this problem~\cite{AS04}.
The recent result of Rosmanis gives the best possible bound with respect to the range of the input function~\cite{ros14}.

\begin{definition}
The Element Distinctness ($ED$) problem takes input $O:[N] \rightarrow [M]$ with the promise that there exists 
a pair $i,j \in [N]$ such that $O(i) = O(j)$. The problem is to output the pair $(i,j)$.
\end{definition}

In this paper, we use the slightly more structured problem known as {\em Claw finding}~\cite{BHT98}.
\begin{definition}
Given two one-to-one functions $F:[N/2] \rightarrow [M]$ and $G:[N/2] \rightarrow [M]$, a claw is a pair $x,y \in [N]$ such that $F(x)=G(y)$.
The Claw finding ($CF$) problem is, on input $F,G$ to return the pair $x,y$ given that there is exactly one.
\end{definition}
It is easy to prove that $CF$ and $ED$ are equivalent up to constant factors. Given an input $O$ for $ED$, an input for $CF$ can be obtained by randomly cutting $O$ into two functions. The probability that the two colliding elements are splited is $1/2$ and running the algorithm for $CF$ a few times on different random cuts is sufficient to find the collision with high probability. Therefore, upper and lower bounds for $ED$ apply similarly to $CF$, up to constant factors.

\begin{theorem}
\label{thm:boundCF}
For $M\geq N$, the quantum query complexity and time complexity of $ED$ and $CF$ are $\Theta(N^{2/3})$.
The most time-efficient algorithm for these problems uses memory $O(N^{2/3})$.
\end{theorem}
The  goal of this section is to study the problem of extracting keys from the double iteration of a block cipher.
In this context, we assume that the quantum cryptanalyst has implemented the publicly known block cipher on
a quantum computer. He then receives the classical data consisting in couples of plaintext ($P$) and ciphertext ($C$), all encrypted with the same key. 
Finally, he uses this datas to extract the secret key with the help of the quantum computer.
We assume that there is only one key that maps $P$ to $C$. Equivalently, we can assume that
the attacker knows a few pairs $(P_i,C_i)$, all encrypted with the same keys.
This ensures that the key mapping $P_i$ to $C_i$ for all $i$ is unique with very high probability.
This can always be simulated by giving only one pair to the attacker and increasing the size of the blocks.
This introduces constant factors in the complexity analysis, and enforces the permutations to have a product structure, but
do not induce any fundamental change to the security proof given here.
Notice that applying this trick implies that $M$ and $N$ are then not comparable anymore, which is an important fact in Section~\ref{sec:4enc}.

\begin{definition}
The 2-Key Extraction ($KE_2^{P,C}$) problem with $P,C \in [M]$ takes input $\mathcal{F}$ where $\mathcal{F}= \{F_1, \ldots, F_N\}$ is a collection of
permutations $F_i \in \mathcal S_{[M]}$  with the promise that there exists a unique couple $(k_1,k_2)$ such that $F_{k_2}(F_{k_1}(P))=C$. The goal of the problem is to output the pair $(k_1,k_2)$.
\end{definition}

\begin{figure}[h]
\begin{center}
\psset{xunit=.5pt,yunit=.5pt,runit=.5pt}
\begin{pspicture}(000,780)(500,1000)
{
\newrgbcolor{curcolor}{0 0 0}
\pscustom[linewidth=1,linecolor=curcolor]
{
\newpath
\moveto(140,980)
\lineto(220,980)
\lineto(220,880)
\lineto(140,880)
\closepath
}
}
{
\newrgbcolor{curcolor}{0 0 0}
\pscustom[linewidth=1,linecolor=curcolor]
{
\newpath
\moveto(280,980)
\lineto(360,980)
\lineto(360,880)
\lineto(280,880)
\closepath
}
}
{
\newrgbcolor{curcolor}{0 0 0}
\pscustom[linewidth=1,linecolor=curcolor]
{
\newpath
\moveto(180,820.00000262)
\lineto(180,880.00000262)
}
}
{
\newrgbcolor{curcolor}{0 0 0}
\pscustom[linewidth=1,linecolor=curcolor]
{
\newpath
\moveto(170,860.00000262)
\lineto(180,880.00000262)
}
}
{
\newrgbcolor{curcolor}{0 0 0}
\pscustom[linewidth=1,linecolor=curcolor]
{
\newpath
\moveto(190,860.00000262)
\lineto(180,880.00000262)
}
}
{
\newrgbcolor{curcolor}{0 0 0}
\pscustom[linewidth=1,linecolor=curcolor]
{
\newpath
\moveto(320,820.00000262)
\lineto(320,880.00000262)
}
}
{
\newrgbcolor{curcolor}{0 0 0}
\pscustom[linewidth=1,linecolor=curcolor]
{
\newpath
\moveto(310,860.00000262)
\lineto(320,880.00000262)
}
}
{
\newrgbcolor{curcolor}{0 0 0}
\pscustom[linewidth=1,linecolor=curcolor]
{
\newpath
\moveto(330,860.00000262)
\lineto(320,880.00000262)
}
}
{
\newrgbcolor{curcolor}{0 0 0}
\pscustom[linewidth=1,linecolor=curcolor]
{
\newpath
\moveto(80,929.99999962)
\lineto(140,929.99999962)
}
}
{
\newrgbcolor{curcolor}{0 0 0}
\pscustom[linewidth=1,linecolor=curcolor]
{
\newpath
\moveto(120,939.99999962)
\lineto(140,929.99999962)
}
}
{
\newrgbcolor{curcolor}{0 0 0}
\pscustom[linewidth=1,linecolor=curcolor]
{
\newpath
\moveto(120,919.99999962)
\lineto(140,929.99999962)
}
}
{
\newrgbcolor{curcolor}{0 0 0}
\pscustom[linewidth=1,linecolor=curcolor]
{
\newpath
\moveto(360,929.99999962)
\lineto(420,929.99999962)
}
}
{
\newrgbcolor{curcolor}{0 0 0}
\pscustom[linewidth=1,linecolor=curcolor]
{
\newpath
\moveto(400,939.99999962)
\lineto(420,929.99999962)
}
}
{
\newrgbcolor{curcolor}{0 0 0}
\pscustom[linewidth=1,linecolor=curcolor]
{
\newpath
\moveto(400,919.99999962)
\lineto(420,929.99999962)
}
}
{
\newrgbcolor{curcolor}{0 0 0}
\pscustom[linewidth=1,linecolor=curcolor]
{
\newpath
\moveto(80,929.99999962)
\lineto(140,929.99999962)
}
}
{
\newrgbcolor{curcolor}{0 0 0}
\pscustom[linewidth=1,linecolor=curcolor]
{
\newpath
\moveto(120,939.99999962)
\lineto(140,929.99999962)
}
}
{
\newrgbcolor{curcolor}{0 0 0}
\pscustom[linewidth=1,linecolor=curcolor]
{
\newpath
\moveto(120,919.99999962)
\lineto(140,929.99999962)
}
}
{
\newrgbcolor{curcolor}{0 0 0}
\pscustom[linewidth=1,linecolor=curcolor]
{
\newpath
\moveto(220,929.99999962)
\lineto(280,929.99999962)
}
}
{
\newrgbcolor{curcolor}{0 0 0}
\pscustom[linewidth=1,linecolor=curcolor]
{
\newpath
\moveto(260,939.99999962)
\lineto(280,929.99999962)
}
}
{
\newrgbcolor{curcolor}{0 0 0}
\pscustom[linewidth=1,linecolor=curcolor]
{
\newpath
\moveto(260,919.99999962)
\lineto(280,929.99999962)
}

\rput(63.54199219,930){$P$}
\rput(440,930){$C$}
\rput(180,805){$k_1$}
\rput(320,805){$k_2$}
\rput(180,930){$F_{k_1}$}}
\rput(320,930){$F_{k_2}$}

\end{pspicture}
\caption{The 2-encryption problem consists in recovering the keys $k_1$ and $k_2$, given a plaintext $P$ and a corresponding ciphertext $C$}
\end{center}
\end{figure}
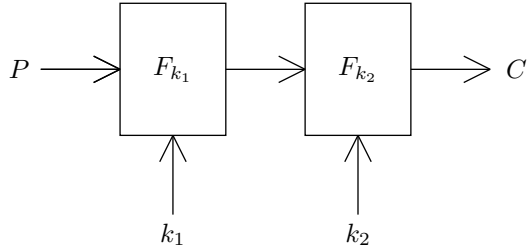

It is easy to prove that the complexity of the problem is independent of the pair $(P,C)$. An algorithm for a given pair $(P,C)$
can be easily adapted to solve the problem for another pair $(P', C')$.
Let $\sigma$ be the permutation that transposes $P'$ and $P$ and $C'$ and $C$.
It suffices to conjugate every permutation of an input $\mathcal F = \{F_1, \ldots, F_N\}$ with $\sigma$ before running the algorithm
for $KE_2^{P,C}$. When it is clear from the context and unnecessary for proofs, we drop the exponent and write only $KE_2$.

The next theorem shows an upper bound on the complexity of $KE_2$. 
The idea is simply to reduce $KE_2$ to~$CF$.

\begin{theorem}
\label{thm:KE2UB}
There exists a quantum algorithm that solves $KE_2$ in time $O(N^{2/3})$ using memory $O(N^{2/3})$.
\end{theorem}

\begin{proof}
Assume that $\mathcal F=\{F_1, \ldots, F_N\}$ is an input of $KE_2^{P,C}$.
We construct the following input to $CF$: a pair of functions $G_1, G_2: [N] \rightarrow [M]$, 
defined by $G_1(x) = F_x(P)$ and $G_2(y) = F_y^{-1}(C)$.
Suppose that the algorithm for $CF$ returns a pair $(x^*,y^*)$. This implies that it is the unique pair such that 
$F_{x^*}(P)=F_{y^*}^{-1}(C)$, leading to $F_{y^*}(F_{x^*}(P))=C$. It is therefore the pair of keys used to encrypt $P$.
\end{proof}

We also prove the following lower bound on the problem.

\begin{theorem}
\label{thm:KE2LB}
A quantum algorithm that solves $KE_2$ needs $\Omega(N^{2/3})$ quantum queries to the input $\mathcal F = \{F_1, \ldots, F_N\}$, including queries to inverse permutations, except with vanishing probabilities.
\end{theorem}

A lower bound on query complexity translates into a
lower bound on time complexity. Therefore, combining the upper and lower bounds on time complexity with the upper
bound on memory gives the following corollary.

\begin{corollary}
\label{cor:TS}
The most time-efficient attack on $KE_2$ has time-space product $O(N^{4/3})$.
\end{corollary}
Notice that a simple Grover search leads to time $O(N)$ with logarithmic space, thus having a time-space product of $N$,
better than the best known algorithm for $CF$.

The rest of the section is now devoted to proving the lower bound on the attack.
There are two important differences between $CF$ and $KE_2$:
\begin{itemize}
\item It is possible to query a permutation $F_i$ on any input $X\in [M]$.
\item It is possible to query inverse permutations $F_i^{-1}$.
\end{itemize}

These differences imply that
there is no obvious query-preserving reduction from $CF$  to $KE_2$.
One important issue is that, since $F_i$ is a permutation, querying $F_i(X)$  with $X \neq P$
gives information on $F_i(P)$. Intuitively, we believe this information is so small that it may not be useful
to solve the problem.
This is even more problematic
in the quantum setting, where a single query can be made in superposition of all the inputs,
preventing the strategy that consists in building the reduction on the fly, when the queries to the input are done.
To overcome this issue, we use a specific tool from quantum query complexity known as the generalized
adversary method.

We prove a slightly stronger lower bound result by considering
the decision version of $KE_2$. In this case, the problem is to determine if there
exists keys $k_1, k_2$ such that $C=F_{k_2}(F_{k_1}(P))$. We denote this version $d{-}KE_2$. 
Of course, an algorithm for the search version can easily be transformed into an algorithm for the decision version. Consequentely, a
lower bound for the decision version is also a lower bound on the search problem.
We compare this algorithm with the decision version of claw finding, denoted $d{-}CF$.
The bounds given in Theorem~\ref{thm:boundCF} apply equivalently to the decision version of these problems.

The proof of Theorem~\ref{thm:KE2LB} is in two steps. In Lemma~\ref{lm:WCLB}, we prove a lower bound in the \emph{worst-case quantum query complexity} using the generalized adversary method. In the second step, which is the proof of Theorem~\ref{thm:KE2LB}, we use a self-reducibility argument to prove that the probability of solving the problem on a random input with less queries than in the worst case is vanishing.

\begin{lemma}
\label{lm:WCLB}
A quantum algorithm that solves $d{-}KE_2$ needs $\Omega(N^{2/3})$ quantum queries to the input $\mathcal F = \{F_1, \ldots, F_N\}$, including queries to inverse permutations.
\end{lemma}

\begin{proof}

We consider an optimal adversary matrix for $d{-}CF$ and use it to build an explicit adversary matrix for $d{-}KE_2$.
Details on the adversary method can be found in the Appendix. In this proof, we need the fact that the generalized adversary
method gives optimal lower bounds on quantum query complexity. 
Therefore, we can choose an
adversary matrix for $CF$ such that the adversary value given by this matrix is 
exactly the quantum query complexity of the function.
Let $\GCF$ be this adversary matrix for $d{-}CF$. The rows and columns of $\GCF$ are indexed by inputs to the problem,
that is, pairs of function $G_1,G_2$.

We now construct an adversary matrix $\GKE$ for the problem $d{-}KE_2^{P,C}$. The rows and columns of $\GKE$ are indexed by collections of permutations $\mathcal F = \{F_1, \ldots F_N\}$. Given a row $u$, (resp. a column $v$), denote $u^{(P,C)}$ the row (resp. $v^{(P,C)}$ the column) of $\GCF$ corresponding to functions $G_1,G_2$ defined by $G_1(k)=F_k(P)$ and $G_2(k)=F_k^{-1}(C)$.
The input $u^{(P,C)}$ for the problem $d{-}CF$ is called the projection of the input $u$ for the problem $d{-}KE_2$. This
operation is represented on figure~\ref{fig:projection}.
We simply define the entries of $\GKE$ by $\GKE[u,v]=\GCF[u^{(P,C)}, v^{(P,C)}]$.
Our goal is now to apply Theorem~\ref{thm:advq}. In order to apply it, we need to compute the values of
$\norm{\GKE}$ and $\max_q \norm{\GKE \bullet \Delta_q}$, where ``$\bullet$'' denotes the entry-wise product of 
matrices\footnote{This product is usually denoted $\circ$, we use here a different notation to avoid confusion with
the composition of permutations that is also used in this proof}.

This definition implies that $\GKE$ has a simple tensor product structure. Notice that for two pairs $(u,v)$ and $(u',v')$ of inputs of $d{-}KE_2$ projecting onto the same pair of inputs $(\tilde u, \tilde v)$ of $CF$, we get by definition $\GKE[u,v]=\GKE[u',v']= \GCF[\tilde u, \tilde v]$. Moreover, the number of inputs to $d{-}KE_2$ projecting onto the same of $d{-}CF$ is a constant. Denoting this constant $D$ and  $\mathbb J_{D \times D}$ the all-one matrix of dimension $D$, we get $\GKE = \GCF \otimes \mathbb J_{D\times D}$.
This immediately gives the relation $$\norm{\GKE} = D \norm{\GCF}.$$ 

The next step is to bound $\max_i{\norm{ \GKE \bullet \Delta_i}}$ where $\Delta_i[u,v] = 0$ if 
$u_i=v_i$, and 1 otherwise. A query to an input to $d{-}KE_2$ is a triplet $(x,k,b)$ where $x \in [M], k \in [N]$ and $b \in \{-1,1\}$ indicates if the query is to a permutation or to its inverse. The query thus returns $F_k^b(x)$.

We show that it is sufficient to consider a special set of queries. The set of queries $\mathcal I$ consists of queries $q=(x,k,b)$ where $x=P$ if $b=1$ and $x=C$ if $b=-1$.
We prove that if the value $\norm{ \GKE \bullet \Delta_q}$ is not
maximized by a query from $\mathcal I$, it is possible to find another matrix with both the same norm and the same tensor product structure such that $\norm{ \Gamma' \bullet \Delta_q}$ is maximized by a query from~$\mathcal I$. Formally, we prove the following claim.

\noindent
{\em Claim:}
There exists a permutation of rows and columns of $\GKE$ leading to a matrix
 $\Gamma'$ such that $\Gamma' =\GCF' \otimes \mathbb J_{D\times D}$, where $\GCF'$ is obtained by permuting the rows and columns of $\GCF$ and
$\max_q{\norm{ \GKE \bullet \Delta_q}}
=\max_{q \in \mathcal I }{\norm{ \Gamma' \bullet \Delta_q}}$.

We first explain how to finish the proof assuming this claim.
A query $q \in \mathcal I$ projects onto a query $\tilde q$ to inputs of  $d{-}CF$. 
Formally, for a query $q \in \mathcal I$, there exists a query $\tilde q$ to inputs of $d{-}CF$
such that for any input $u$ of $d{-}KE_2$, $u_q =  u^{(P,C)} _{\tilde q}$.
If $q=(x,k,b)\in \mathcal I$, the query $\tilde q = (k,b)$ on $u^{(P,C)}$ returns $G_1(k)$ if $b=1$ and $G_2(k)$ if $b=-1$. 
By definition of $u^{(P,C)}$, we have $G_1(k)=F_k(P)$ and $G_2(k)=F_k^{-1}(C)$ and thus,  $u_q =  u^{(P,C)} _{\tilde q}$. 
This implies that for $q\in \mathcal I$, 
$(\Gamma' \bullet \Delta_q) [u,v] = (\GCF'\bullet \Delta_{\tilde q})[u^{(P,C)},v^{(P,C)}]$. The tensor product structure ensures that
$\Gamma' \bullet \Delta_q = (\GCF' \otimes \mathbb J_{D\times D}) \bullet \Delta_{q} = (\GCF' \bullet \Delta_{\tilde q})\otimes
\mathbb J_{D\times D}$, and thus

$$\max_q{\norm{ \GKE \bullet \Delta_i}} = \max_{q \in \mathcal I}{\norm{ \Gamma' \bullet \Delta_q}} = D \max_{\tilde q}\norm{\GCF' \bullet \Delta_{\tilde q}}.$$
The last equality is true because since $\mathcal I$ and queries to inputs of $d{-}CF$ have the same cardinality, they are in one-to-one correspondance. Maximizing over queries to $\mathcal I$ is therefore equivalent to maximizing over queries to inputs to $d{-}CF$.
Finally, this shows that the quantum query complexity of $KE_2$ is at least

\begin{eqnarray*}
\min_q \frac{\norm{\GKE}}{\norm{\GKE \bullet \Delta_q}} &=& \min_q \frac{\norm{\Gamma'}}{\norm{\Gamma' \bullet \Delta_q}} = \min_{\tilde q} \frac{D\norm{\GCF'}}{D\norm{\GCF' \bullet \Delta_{\tilde q}}} \\
&=& \min_{\tilde q} \frac{\norm{\GCF}}{\norm{\GCF \bullet \Delta_{\tilde q}}} = \Omega(N^{2/3}).
\end{eqnarray*}

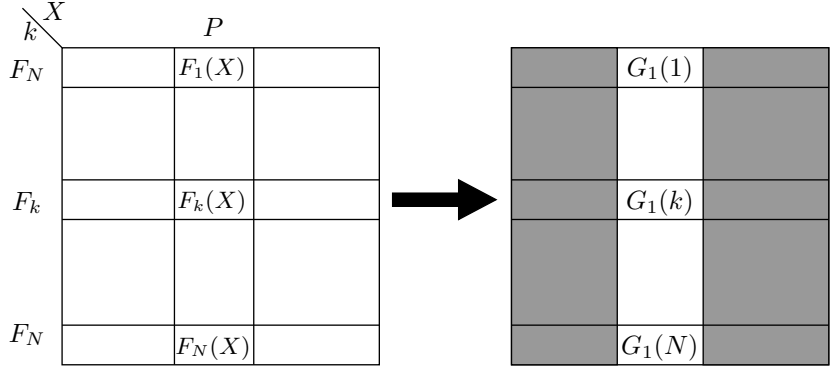
\begin{figure}[h]
\begin{center}
\psset{xunit=.5pt,yunit=.5pt,runit=.5pt}
\begin{pspicture}(620,300)
{
\newrgbcolor{curcolor}{0 0 0}
\pscustom[linewidth=1.01058233,linecolor=curcolor]
{
\newpath
\moveto(35.833328,257.16666)
\lineto(275.833328,257.16666)
\lineto(275.833328,17.16666)
\lineto(35.833328,17.16666)
\closepath
}
}
{
\newrgbcolor{curcolor}{0 0 0}
\pscustom[linewidth=1,linecolor=curcolor]
{
\newpath
\moveto(35.833328,127.16666)
\lineto(275.833328,127.16666)
}
}
{
\newrgbcolor{curcolor}{0 0 0}
\pscustom[linewidth=1,linecolor=curcolor]
{
\newpath
\moveto(120.833328,257.16666)
\lineto(120.833328,17.16666)
}
}
{
\newrgbcolor{curcolor}{0 0 0}
\pscustom[linewidth=1,linecolor=curcolor]
{
\newpath
\moveto(180.833328,257.16666)
\lineto(180.833328,17.16666)
}
}
{
\newrgbcolor{curcolor}{0 0 0}
\pscustom[linewidth=1,linecolor=curcolor]
{
\newpath
\moveto(35.833328,257.16666)
\lineto(5.833328,287.16666)
}
}
\rput(30,285){$X$}
\rput(12,270){$k$}
\rput(9.5,140){$F_k$}
\rput(9.5,40){$F_N$}
\rput(9.5,240){$F_N$}
\rput(150,273){$P$}
\rput(149,140){{\small $F_k(X)$}}
\rput(149,241){{\small $F_1(X)$}}
\rput(150,30){{\small $F_N(X)$}}
\rput(488,140){$G_1(k)$}
\rput(488,241){$G_1(1)$}
\rput(488,30){$G_1(N)$}
{
\newrgbcolor{curcolor}{0 0 0}
\pscustom[linestyle=none,fillstyle=solid,fillcolor=curcolor]
{
\newpath
\moveto(285.833328,147.16666)
\lineto(335.833328,147.16666)
\lineto(335.833328,137.16666)
\lineto(285.833328,137.16666)
\closepath
}
}
{
\newrgbcolor{curcolor}{0 0 0}
\pscustom[linewidth=0.91287094,linecolor=curcolor]
{
\newpath
\moveto(285.833328,147.16666)
\lineto(335.833328,147.16666)
\lineto(335.833328,137.16666)
\lineto(285.833328,137.16666)
\closepath
}
}
{
\newrgbcolor{curcolor}{0 0 0}
\pscustom[linewidth=1.01058233,linecolor=curcolor]
{
\newpath
\moveto(375.833328,257.16666)
\lineto(615.833328,257.16666)
\lineto(615.833328,17.16666)
\lineto(375.833328,17.16666)
\closepath
}
}
{
\newrgbcolor{curcolor}{0 0 0}
\pscustom[linestyle=none,fillstyle=solid,fillcolor=curcolor]
{
\newpath
\moveto(335.833328,157.16666)
\lineto(335.833328,127.16666)
\lineto(364.166658,142.16666)
\closepath
}
}
{
\newrgbcolor{curcolor}{0 0 0}
\pscustom[linewidth=1,linecolor=curcolor]
{
\newpath
\moveto(335.833328,157.16666)
\lineto(335.833328,127.16666)
\lineto(364.166658,142.16666)
\closepath
}
}
{
\newrgbcolor{curcolor}{0.60000002 0.60000002 0.60000002}
\pscustom[linestyle=none,fillstyle=solid,fillcolor=curcolor]
{
\newpath
\moveto(375.833328,257.16666)
\lineto(455.833328,257.16666)
\lineto(455.833328,17.16666)
\lineto(375.833328,17.16666)
\closepath
}
}
{
\newrgbcolor{curcolor}{0 0 0}
\pscustom[linewidth=1.02150786,linecolor=curcolor]
{
\newpath
\moveto(375.833328,257.16666)
\lineto(455.833328,257.16666)
\lineto(455.833328,17.16666)
\lineto(375.833328,17.16666)
\closepath
}
}
{
\newrgbcolor{curcolor}{0.60000002 0.60000002 0.60000002}
\pscustom[linestyle=none,fillstyle=solid,fillcolor=curcolor]
{
\newpath
\moveto(520.833328,257.16666)
\lineto(615.833328,257.16666)
\lineto(615.833328,17.16666)
\lineto(520.833328,17.16666)
\closepath
}
}
{
\newrgbcolor{curcolor}{0 0 0}
\pscustom[linewidth=1.02150786,linecolor=curcolor]
{
\newpath
\moveto(520.833328,257.16666)
\lineto(615.833328,257.16666)
\lineto(615.833328,17.16666)
\lineto(520.833328,17.16666)
\closepath
}
}
{
\newrgbcolor{curcolor}{0 0 0}
\pscustom[linewidth=1,linecolor=curcolor]
{
\newpath
\moveto(35.833328,157.16666)
\lineto(275.833328,157.16666)
}
}
{
\newrgbcolor{curcolor}{0 0 0}
\pscustom[linewidth=1,linecolor=curcolor]
{
\newpath
\moveto(35.833328,227.16666)
\lineto(275.833328,227.16666)
}
}
{
\newrgbcolor{curcolor}{0 0 0}
\pscustom[linewidth=1,linecolor=curcolor]
{
\newpath
\moveto(35.833328,47.16666)
\lineto(275.833328,47.16666)
}
}
{
\newrgbcolor{curcolor}{0 0 0}
\pscustom[linewidth=1,linecolor=curcolor]
{
\newpath
\moveto(375.833328,127.16666)
\lineto(615.833328,127.16666)
}
}
{
\newrgbcolor{curcolor}{0 0 0}
\pscustom[linewidth=1,linecolor=curcolor]
{
\newpath
\moveto(375.833328,157.16666)
\lineto(615.833328,157.16666)
}
}
{
\newrgbcolor{curcolor}{0 0 0}
\pscustom[linewidth=1,linecolor=curcolor]
{
\newpath
\moveto(375.833328,227.16666)
\lineto(615.833328,227.16666)
}
}
{
\newrgbcolor{curcolor}{0 0 0}
\pscustom[linewidth=1,linecolor=curcolor]
{
\newpath
\moveto(375.833328,47.16666)
\lineto(615.833328,47.16666)
}
}
\end{pspicture}
\caption{\label{fig:projection}An input to $KE_2$ can be represented as a table whose rows and columns are
indexed by $k$ and $X$, respectively. Each line of the table is a permutation $F_i$.
The projecting onto an input to $CF$ is a restriction to one column.
To build a complete input to $CF$, one also has to restrict the table of permutations $F^{-1}_k$.}
\end{center}
\end{figure}

It only remains to prove the above claim.
Suppose that $\norm{\GKE \bullet \Delta_q}$ is maximized by
a query $q^*=(x^*,k^*,b^*)$. The intuition of the claim is that if
$q^* \notin \mathcal I$, it can still be projected onto a query to inputs of $d{-}CF$. By mapping these to inputs of the original problems,
we get the new matrix $\Gamma'$.

Assume that $b^*=1$ and $x^* \neq P$ (the proof is similar with $b^*=-1$).
Let $\sigma$ denote the transposition $(x^* \ P)$.
For an input u = $\{F_i\}_{i \in [N]}$, denote $u ^ \sigma = \{F_i\circ \sigma\}_{i \in [N]}$ and define $\Gamma'$ as $\Gamma'[u,v] = \GKE [u^\sigma, v^\sigma]$.
The operation $u \mapsto u^\sigma$ corresponds to a permutations of rows and columns
of $\GKE$. Denote $\Pi_\sigma$ this permutation, so that $\Gamma' = \Pi_\sigma^\dagger \GKE\Pi_\sigma$.
Similarly, we have
$\Pi_\sigma^\dagger \Delta_{q^*}\Pi_\sigma = \Delta_{q^{**}}$,
where $q^{**}=(P,k^*,b^*)$.

Finally, we show that $\Gamma' = \GCF' \otimes \mathbb J_{D\times D}$ 
for some matrix $\GCF'$. 
The sets $\{u^{(P,C)}\}_{u}$ and $\{u^{(x^*,C)}\}_{u}$ are equal and therefore, there exists a bijection $\tau$ sending
$u^{(P,C)}$ to $u^{(x,C)}$.
This bijection satisfies $(u^\sigma)^{(P,C)} = \tau(u^{(P,C)})$.
The matrix $\GCF'$ is simply defined as $\GCF'[\tilde u, \tilde v]=\GCF[\tau(\tilde u),\tau(\tilde v)]$.
This gives
\begin{eqnarray*}
\Gamma'[u,v] &=& \GKE[u^\sigma,v^\sigma] = \GCF[(u^\sigma)^{(P,C)}, (v^\sigma)^{(P,C)}],\\
	& =& \GCF[\tau(u^{(P,C)}), \tau(u^{(P,C)})] = \GCF'[u^{(P,C)},v^{(P,C)}].
\end{eqnarray*}

\hide{
The operation $\tilde u \mapsto \tau(\tilde u)$ again corresponds to a permutation of the rows and columns, denoted $\Pi_\tau$.
We finally get 
\begin{eqnarray*}
\Gamma' &=& \Pi_\sigma^\dagger \GKE \Pi_\sigma = \Pi_\sigma^\dagger (\GCF \otimes \mathbb J_{D\times D})\Pi_\sigma,\\
	&=& (\Pi_\tau^\dagger \GCF \Pi_\tau )\otimes \mathbb J_{D\times D} = \GCF' \otimes \mathbb J_{D \times D}.
\end{eqnarray*}
}
This immediately leads to $\Gamma' = \GCF' \otimes \mathbb J_{D \times D}$, and finishes the proof of the claim.
\end{proof}

\begin{proof}[Proof of Theorem~\ref{thm:KE2LB}]
We now show that the lower bound proved in Lemma~\ref{lm:WCLB} holds on random inputs, except with vanishing probability.
The proof is in three steps. In the first step, we prove that the lower bound holds for the average number of queries made by the attack. In the second step, we show that the lower bound also holds when considering the average error.
Finally, we show that the lower bound hold for random inputs, except with vanishing probability.

\noindent
{\bf Step 1:} Let $\mathcal A$ be an attack such that, given an input $\mathcal F =\{F_1, \ldots, F_N\}$ to $KE_2^{P,C}$,
returns $(k_1,k_2)$ such that $F_{k_2}(F_{k_1}(P))=C$ after $q$ queries on average over the inputs and is successful with probability at least $1-\eps$ for any input.
Consider then the following attack: run $\mathcal A$ for $q/\eps$ queries. If $\mathcal A$ stopped, then output the same value. Otherwise output a random value. By Markov's inequality, this new attach is successful with probability at least $1-2\eps$ and therefore,
by Lemma~\ref{lm:WCLB}, $q/\eps \geq \Omega(N^{2/3})$.

\noindent
{\bf Step 2:} Let $\mathcal A'$ be an attack such that, given an input $\mathcal F =\{F_1, \ldots, F_t\}$ to $KE_2^{P,C}$,
returns $(k_1,k_2)$ such that $F_{k_2}(F_{k_1}(P))=C$ after $q$ queries on average and is successful with average probability $1-\eps$, both over the inputs.
Consider now the following attack: Choose $N$ random permutations $\{\sigma_1, \ldots \sigma_N\}$, and run
the $\mathcal A'$ on the input $\mathcal F' = \{\sigma_1 \circ F_1, \ldots, \sigma_N \circ F_N\}$. This is equivalent
to running $\mathcal A'$ on a random input, so that the error made by the attack is now $1-\eps$ for any input.
From Step~1, we get again $q = \Omega(N^{2/3})$.

\noindent
{\bf Step 3:} Let $\mathcal A''$ be an attack that solves $KE_2^{P,C}$ with error $\eps$ on average. Denote $Q$ the
random variable indicating the number of queries made by $\mathcal A''$ and fix $q=o(N^{2/3})$.
Denote $\delta = \Pr[Q\leq q]$. We want to prove that $\delta$ is vanishing.
Fix two constants $k$ and $k'$ and consider the following attack. Repeat $k/\delta$ times the following steps:
\begin{enumerate}
\item Choose $N$ random permutation $\sigma_1, \ldots, \sigma_N$. Run $\mathcal A''$ on a random input as explained
in Step~2, and stop it after $k'q$ queries.
\item If $\mathcal A''$ stopped, then output the same value and stop.
\end{enumerate}
If after repeating these two steps $k/\delta$ times, no output was produced, output random values.

The number of queries of this attack is at most $kk'q/\delta$. Choosing $k$ and $k'$ large enough, the probability that at least one iteration of the loop is successful can be made arbitrarily close to one, so that the error probability of this attack is arbitrarily close to
$\eps$. We have constructed an attack that makes $O(q/\delta)$ queries on average and is successful with average probability $1-\eps$.
From Step~2, we get that $q/\delta = \Omega(N^{2/3})$, which implies $\delta = o(1)$.

\end{proof}

In order to measure the gains obtained by quantum algorithms, we use a quantity similar to the one used by 
Dinur, Dunkelman, Keller and Shamir~\cite{ddks12}. We consider 
$\log{C}/\log{Q}$, where $C$ is a classical complexity measure and $Q$ is its quantum counterpart.
Intuitively, this corresponds to the factor by which the key size is decreased
when the attack is quantized.
For example the gain in time for Grover search over classical exhaustive search is 2.
The gain for the time-space product is similar because both Grover and exhaustive search require only constant space.

The gain for the $\MITM$ attack is very different. The gain in time is $3/2$, using the algorithm presented above.
Since this algorithm is optimal, it is the largest possible gain.
The gain for the time-space product of this algorithm is also $3/2$.
Interestingly, there exists other algorithms for $ED$ leading to different gains. For example,
the algorithm based on amplitude amplification of~\cite{ED05} leads to a gain in time of $4/3$ (over
the classical $\MITM$ attack), and a gain in time-space product of 8/5, better than the most time-efficient attack.
The most time-efficient algorithm is not the one leading to the most important gain in
time-space, and an attacker that is willing to pay with more time can save on the time-space product.

\begin{table}[h]
\begin{center}
\begin{tabular}{|c|c|c|}
\cline{2-3} \multicolumn{1}{c|}{}	& Time & Time-space \\
   \hline
Exhaustive search& $2$ & $2$ \\
   \hline
MITM & $3/2=1.5$ & $3/2=1.5$ \\
   \hline
Amplitude amplification~\cite{ED05} & $4/3\simeq 1.3$ & $8/5=1.6$\\
   \hline   
\end{tabular}
\caption{\label{table:2EK} Gains of attacks against 4-encryption}
\end{center}
\end{table}

\section{Quantum attack against $4$-encryption}
\label{sec:4enc}

We now apply tools from quantum complexity theory to study the case of four iterated encryptions.
We assume again that 
the attacker knows sufficiently many 
pairs $(P_i,C_i)$ of plaintextcipher-text in order to 
ensure
that, with very high probability, there is only one quadruple of keys
$(k_1, k_2, k_3, k_4)$ that satisfies the relations $C_i=E_{k_4}(E_{k_3}(E_{k_2}(E_{k_1}(P_i))))$ for all $i$.

The classical $\MITM$ attack can be applied in this situation by considering the pairs $(k_1,k_2)$ and $(k_2,k_3)$
as single keys. This requires time and memory $O(N^2)$, and thus time-space product $O(N^4)$.
This was the best known algorithm the recent dissection attack~\cite{ddks12}, which still
requires time $O(N^2)$ but memory only $O(N)$. This significantly improves the time-space product
to $O(N^3)$.

The basic idea of the dissection attack is to make an exhaustive search of the intermediate value after two encryptions. The candidate
values are then checked using a $\MITM$ procedure.
A notable difference with the previous case is that the complexity of the problem is now a function of both $M$ and $N$.
We assumed in the beginning that $M$ and $N$ were of comparable sizes. 
In order to keep this assumption, we cannot assume anymore that the attacker has a single pair of datas. Instead, we
assume that it has enough pairs, all encrypted with the same datas, to be sure that, with large probability, there is only one quadruple of keys consistant with all the datas. It can be shown that four pairs of plaintext with corresponding ciphertext.

\begin{definition}
The 4-Key Extraction ($KE_4^{P,C}$) problem with $P=(P_1,P_2,P_3,P_4) ,C=(C_1,C_2,C_3,C_4) \in [M]^4$ takes input $\mathcal{F}$ where $\mathcal{F}= \{F_1, \ldots, F_N\}$ is a collection of
permutations of $[M]$. The goal of the problem is to output $(k_1,k_2,k_3,k_4)$ such that $F_{k_4}(F_{k_3}(F_{k_2}(F_{k_1}(P_i))))=C_i$ for all $i$.
\end{definition}

The attack uses the quantum $\MITM$ algorithm presented in the previous section as
a subroutine. The basic idea is to compose this algorithm with a Grover search of the value
in the middle. Quantum query complexity has the remarkable property, derived from the generalized adversary method, that
for compatible functions $f$ and $g$, $Q(f(g(x^1), \ldots, g(x^n))) = O(Q(f) Q(g))$~\cite{hls07}.

For simplicity, consider first the decision version of the problem.
Given permutations $\mathcal F = \{F_1,\ldots, F_N\}$ and a pair $(P, C)$, the problem $d{-}KE_4$ is to decide if there exists a 4-tuple of keys $({k_1}, {k_2}, {k_3}, {k_4})$ such that $C = F_{k_4}(F_{k_3}(F_{k_2}(F_{k_1}(P))))$.

Consider the following functions:
\begin{eqnarray*}
f_X:  (\mathcal S_{[M]})^N & \rightarrow & \{0,1\}\\
	\{F_1, \ldots, F_N\} & \mapsto & \begin{cases}1 & \text{if there exists $k_1$ and $k_2$ such that $F_{k_2}(F_{k_1}(P))=X$,}\\
					    0 & \text{otherwise;}
					    \end{cases}
\end{eqnarray*}
and
\begin{eqnarray*}
g_X:(\mathcal S_{[M]})^N & \rightarrow & \{0,1\}\\
\{F_1, \ldots, F_N\} & \mapsto & \begin{cases}1 & \text{if there exists $k_1$ and $k_2$ such that $F_{k_2}(F_{k_1}(X))=C$,}\\
					    0 & \text{otherwise.}
					    \end{cases}
\end{eqnarray*}
The function $f_X$ returns the boolean answer of the problem $d{-}KE_2^{(P,X)}$ and $g_X$ the answer of
$d{-}KE_2^{(X,C)}$.
Finally, denote $\search$ the function that, on input $u \in \{0,1\}^M$, returns 1 if there exists $i$ such that $u_i = 1$ and 0 otherwise.

The problem $d{-}KE_4: (\mathcal S_{[M]})^N \rightarrow \{0,1\}$ can be expressed as the following composition:
\begin{equation}
\label{eqn:product}
d{-}KE_4 = \search \circ (f_0 \wedge g_0, f_1 \wedge g_1, \ldots, f_N \wedge g_N).
\end{equation}
It is well known that a quantum query complexity is multiplicative under function composition~\cite{hls07},
which leads to a $O(M^{1/2}N^{2/3})$ upper bound on the number of queries needed to solve the problem.
However, the composition theorem that proves this upper bound holds for \emph{quantum query} complexity, whereas
we are interested here in bounding the \emph{time} and \emph{space} used by a quantum algorithm.
For this purpose, we give an explicit composed algorithm.
This algorithm is a quantized version of the $Dissect_2(4,1)$ algorithm of Dinur, Dunkleman, Keller and Shamir~\cite{ddks12}.
In the classical setting, this algorithm achieves the best known time-space product for 4-encryption.
We now remove the assumption that the key is unique, and suppose that the
attacker has four pairs $(P_i, C_i)$.
We also assume that $M$ and $N$ are of comparable sizes and state our result as a function of $N$.

\begin{theorem}
\label{thm:KE4UB}
There exists a quantum algorithm that solves $KE_4$ in time $O(N^{7/6})$ and using memory $O(N^{2/3})$.
The time-space product for this attack is $O(N^{11/6})$.
\end{theorem}

\begin{proof}
Assume that the attacker knows the pairs $(P_i,C_i)_{i=1..4}$. This is sufficient to ensure that
only a single quadruple of keys is consistent with the datas.
We devise a search algorithm that consists in a composition of a Grover search with
a quantum $\MITM$ attack presented in Section~\ref{sec:qmitm}.
We use the framework of quantum walks in order to compose these two
procedures. The details of the theorems that we use are given in the Appendix.

The walk is designed on the complete graph, whose vertices are indexed by elements from $[M]$
who are candidates for the value $X$ after two encryptions.
Once the correct value for $X$ is found, it suffices to make $O(N^{2/3})$ queries to the input in
order to find the keys, using some of the datas $(P_i,C_i)$. This correct value is unique, and
the goal of the quantum walk is to find it.

Theorem~\ref{thm:MNRS} given in the Appendix states that
the cost to find the correct value $X$ is upper bounded by
$\SSS+\frac{1}{\sqrt \eps}(\frac 1 {\sqrt \delta} \UU + \CC),$
where 
\begin{itemize}
\item $\SSS$ is the cost for setting up a quantum register in a state that corresponds to the stationary
of the classical random walk on the graph
distribution
\item $\UU$ is the cost of moving from one node to an adjacent one,
\item $\CC$ is the cost of checking if the value $X$ is the correct one,
\item $\eps$ is the probability of finding the correct value~$X$, and
\item  $\delta$ is the eigenvalue gap for the complete graph.
\end{itemize}


In our case, the 
setup is to build a uniform superposition of all states $\ket X$, which can
be made with no query to the input, and constant time. The update is the same procedure.
The checking consists in running
a quantum $\MITM$ attack several times in order to check that the value $X$ satisfies
$E_{k_2}(E_{k_1}(P_i))=X$
and $E_{k_4}(E_{k_3}(X))=C_i$ for all $i$ for some value $k_1, k_2, k_3, k_4$.
The probability of finding the correct value is $\eps = 1/M$.
Finally, the eigenvalue gap for the complete graph is $\delta = 1-\frac 1 {M-1}$.
Overall, this leads to an attack that makes 
 $O(M^{1/2}N^{2/3})$ queries, and takes the same amount of time.
Moreover, the checking procedure is the only procedure using non-constant memory.
The attack uses in total $O(N^{2/3})$ memory, leading to a time-space product bounded by $O(M^{1/2}N^{4/3})$.

\begin{figure}[h]
\begin{center}
\psset{xunit=.5pt,yunit=.5pt,runit=.5pt}
\begin{pspicture}(0,530)(780,750)
{
\newrgbcolor{curcolor}{0 0 0}
\pscustom[linewidth=1,linecolor=curcolor]
{
\newpath
\moveto(100,720)
\lineto(180,720)
\lineto(180,620)
\lineto(100,620)
\closepath
}
}
{
\newrgbcolor{curcolor}{0 0 0}
\pscustom[linewidth=1,linecolor=curcolor]
{
\newpath
\moveto(240,720)
\lineto(320,720)
\lineto(320,620)
\lineto(240,620)
\closepath
}
}
{
\newrgbcolor{curcolor}{0 0 0}
\pscustom[linewidth=1,linecolor=curcolor]
{
\newpath
\moveto(140,560.00000262)
\lineto(140,620.00000262)
}
}
{
\newrgbcolor{curcolor}{0 0 0}
\pscustom[linewidth=1,linecolor=curcolor]
{
\newpath
\moveto(130,600.00000262)
\lineto(140,620.00000262)
}
}
{
\newrgbcolor{curcolor}{0 0 0}
\pscustom[linewidth=1,linecolor=curcolor]
{
\newpath
\moveto(150,600.00000262)
\lineto(140,620.00000262)
}
}
{
\newrgbcolor{curcolor}{0 0 0}
\pscustom[linewidth=1,linecolor=curcolor]
{
\newpath
\moveto(280,560.00000262)
\lineto(280,620.00000262)
}
}
{
\newrgbcolor{curcolor}{0 0 0}
\pscustom[linewidth=1,linecolor=curcolor]
{
\newpath
\moveto(270,600.00000262)
\lineto(280,620.00000262)
}
}
{
\newrgbcolor{curcolor}{0 0 0}
\pscustom[linewidth=1,linecolor=curcolor]
{
\newpath
\moveto(290,600.00000262)
\lineto(280,620.00000262)
}
}
{
\newrgbcolor{curcolor}{0 0 0}
\pscustom[linewidth=1,linecolor=curcolor]
{
\newpath
\moveto(40,670.00000262)
\lineto(100,670.00000262)
}
}
{
\newrgbcolor{curcolor}{0 0 0}
\pscustom[linewidth=1,linecolor=curcolor]
{
\newpath
\moveto(80,680.00000262)
\lineto(100,670.00000262)
}
}
{
\newrgbcolor{curcolor}{0 0 0}
\pscustom[linewidth=1,linecolor=curcolor]
{
\newpath
\moveto(80,660.00000262)
\lineto(100,670.00000262)
}
}
{
\newrgbcolor{curcolor}{0 0 0}
\pscustom[linewidth=1,linecolor=curcolor]
{
\newpath
\moveto(320,670.00000262)
\lineto(380,670.00000262)
}
}
{
\newrgbcolor{curcolor}{0 0 0}
\pscustom[linewidth=1,linecolor=curcolor]
{
\newpath
\moveto(360,680.00000262)
\lineto(380,670.00000262)
}
}
{
\newrgbcolor{curcolor}{0 0 0}
\pscustom[linewidth=1,linecolor=curcolor]
{
\newpath
\moveto(360,660.00000262)
\lineto(380,670.00000262)
}
}
{
\newrgbcolor{curcolor}{0 0 0}
\pscustom[linewidth=1,linecolor=curcolor]
{
\newpath
\moveto(40,670.00000262)
\lineto(100,670.00000262)
}
}
{
\newrgbcolor{curcolor}{0 0 0}
\pscustom[linewidth=1,linecolor=curcolor]
{
\newpath
\moveto(80,680.00000262)
\lineto(100,670.00000262)
}
}
{
\newrgbcolor{curcolor}{0 0 0}
\pscustom[linewidth=1,linecolor=curcolor]
{
\newpath
\moveto(80,660.00000262)
\lineto(100,670.00000262)
}
}
{
\newrgbcolor{curcolor}{0 0 0}
\pscustom[linewidth=1,linecolor=curcolor]
{
\newpath
\moveto(180,670.00000262)
\lineto(240,670.00000262)
}
}
{
\newrgbcolor{curcolor}{0 0 0}
\pscustom[linewidth=1,linecolor=curcolor]
{
\newpath
\moveto(220,680.00000262)
\lineto(240,670.00000262)
}
}
{
\newrgbcolor{curcolor}{0 0 0}
\pscustom[linewidth=1,linecolor=curcolor]
{
\newpath
\moveto(220,660.00000262)
\lineto(240,670.00000262)
}
}
{
\newrgbcolor{curcolor}{0 0 0}
\pscustom[linewidth=1,linecolor=curcolor]
{
\newpath
\moveto(380,720)
\lineto(460,720)
\lineto(460,620)
\lineto(380,620)
\closepath
}
}
{
\newrgbcolor{curcolor}{0 0 0}
\pscustom[linewidth=1,linecolor=curcolor]
{
\newpath
\moveto(520,720)
\lineto(600,720)
\lineto(600,620)
\lineto(520,620)
\closepath
}
}
{
\newrgbcolor{curcolor}{0 0 0}
\pscustom[linewidth=1,linecolor=curcolor]
{
\newpath
\moveto(420,560.00000262)
\lineto(420,620.00000262)
}
}
{
\newrgbcolor{curcolor}{0 0 0}
\pscustom[linewidth=1,linecolor=curcolor]
{
\newpath
\moveto(410,600.00000262)
\lineto(420,620.00000262)
}
}
{
\newrgbcolor{curcolor}{0 0 0}
\pscustom[linewidth=1,linecolor=curcolor]
{
\newpath
\moveto(430,600.00000262)
\lineto(420,620.00000262)
}
}
{
\newrgbcolor{curcolor}{0 0 0}
\pscustom[linewidth=1,linecolor=curcolor]
{
\newpath
\moveto(560,560.00000262)
\lineto(560,620.00000262)
}
}
{
\newrgbcolor{curcolor}{0 0 0}
\pscustom[linewidth=1,linecolor=curcolor]
{
\newpath
\moveto(550,600.00000262)
\lineto(560,620.00000262)
}
}
{
\newrgbcolor{curcolor}{0 0 0}
\pscustom[linewidth=1,linecolor=curcolor]
{
\newpath
\moveto(570,600.00000262)
\lineto(560,620.00000262)
}
}
{
\newrgbcolor{curcolor}{0 0 0}
\pscustom[linewidth=1,linecolor=curcolor]
{
\newpath
\moveto(320,670.00000262)
\lineto(380,670.00000262)
}
}
{
\newrgbcolor{curcolor}{0 0 0}
\pscustom[linewidth=1,linecolor=curcolor]
{
\newpath
\moveto(360,680.00000262)
\lineto(380,670.00000262)
}
}
{
\newrgbcolor{curcolor}{0 0 0}
\pscustom[linewidth=1,linecolor=curcolor]
{
\newpath
\moveto(360,660.00000262)
\lineto(380,670.00000262)
}
}
{
\newrgbcolor{curcolor}{0 0 0}
\pscustom[linewidth=1,linecolor=curcolor]
{
\newpath
\moveto(600,670.00000262)
\lineto(660,670.00000262)
}
}
{
\newrgbcolor{curcolor}{0 0 0}
\pscustom[linewidth=1,linecolor=curcolor]
{
\newpath
\moveto(640,680.00000262)
\lineto(660,670.00000262)
}
}
{
\newrgbcolor{curcolor}{0 0 0}
\pscustom[linewidth=1,linecolor=curcolor]
{
\newpath
\moveto(640,660.00000262)
\lineto(660,670.00000262)
}
}
{
\newrgbcolor{curcolor}{0 0 0}
\pscustom[linewidth=1,linecolor=curcolor]
{
\newpath
\moveto(460,670.00000262)
\lineto(520,670.00000262)
}
}
{
\newrgbcolor{curcolor}{0 0 0}
\pscustom[linewidth=1,linecolor=curcolor]
{
\newpath
\moveto(500,680.00000262)
\lineto(520,670.00000262)
}
}
{
\newrgbcolor{curcolor}{0 0 0}
\pscustom[linewidth=1,linecolor=curcolor]
{
\newpath
\moveto(500,660.00000262)
\lineto(520,670.00000262)
}
}
\rput(23,675){$P$}
\rput(680,675){$C$}
\rput(140,545){$k_1$}
\rput(279,545){$k_2$}
\rput(420,545){$k_3$}
\rput(559,545){$k_4$}
\rput(140,675){$F_{k_1}$}
\rput(279,675){$F_{k_2}$}
\rput(420,675){$F_{k_3}$}
\rput(559,675){$F_{k_4}$}
\rput(347,740){$X$}
\qdisk(345,670){2.5pt}
\pscustom[linewidth=1,linecolor=curcolor]
{
\newpath
\moveto(345,670)
\lineto(345,730)
}

\end{pspicture}
\caption{The attack on 4-encryption is an exhaustive search over the central value $X$ combined with
a $MITM$ algorithm to check the consistency with the datas $P,C$}
\end{center}
\end{figure}
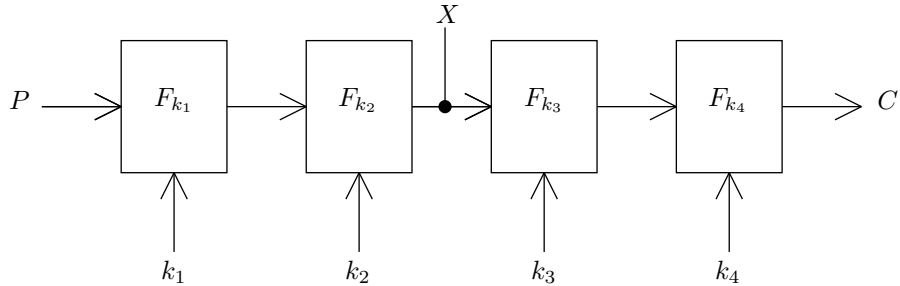

\end{proof}

Our approach does not give any lower bound on the number of queries to the input.
Since our attack is a composition of Grover's algorithm with a quantum $\MITM$ attack,
expressing the problem as a composition could allow to apply product theorems on the
adversary method~\cite{hls07,LMRSS10,BHKKL11}, leading to a lower bound.
Unfortunately, known theorems do not apply to the composition
given in Equation~\ref{eqn:product}.
For example, the composition theorem proven by H\o yer, Lee and \v Spalek
requires the inner function to be boolean and to act on independent inputs.
In our case, all the inner functions act on the same input.
Equivalently, the expression given in Equation~\ref{eqn:product} could be completed
with a non-boolean function that maps the input to $N$ copies of it.
The hypothesis of the composition theorem used to prove security of Merkle puzzles in a
quantum world are very contrived an do not apply here~\cite{BHKKL11}.

We compare three different attacks against 4-encryption:
\begin{itemize}
\item exhaustive search over the whole key-space,
\item classical and quantum $\MITM$ attack where both the first and the second pair of keys are treated as single keys,
\item the dissection attack and its quantized version.
\end{itemize}
The quantization of exhaustive search gives a very good gain, but its time-performance is poor.
Applying the $\MITM$ attack from the previous section gives both poor gains and poor performances.
The quantization of the dissection attack gives better gains than what we observed for 2-encryption, but
we have no indications that this algorithm is optimal.

\begin{table}[h]
\begin{center}
\begin{tabular}{|c|c|c|}
\cline{2-3} \multicolumn{1}{c|}{}	& Time & Time-space \\
   \hline
Exhaustive search& $2$ & $2$ \\
   \hline
MITM & $3/2 =1.5$ & $3/2=1.5$ \\
   \hline
Dissection & $12/7 \simeq 1.7$ & $18/11 \simeq 1.63$\\
   \hline   
\end{tabular}
\caption{\label{table:4EK} Gains of attacks against 4-encryption}
\end{center}
\end{table}

\section{Conclusion}

Iterative encryption is  a natural way to amplify the security of block ciphers.
It is well known that this procedure do not lead to the expected results, even for classical attacks
in the case of 2-encryption. The $\MITM$ attack allows an attacker to save a lot of time,
if he is willing to pay with more space.

Equipped with a quantum computer, an attacker can run a quantum algorithm for collision finding.
We have proven that this approach leads to optimal time complexity for extracting keys, making crucial
use of the generalized adversary method, a tool from quantum complexity theory.
Extracting a pair of keys $(k_1,k_2)$ takes time $N^{2/3}$, where $N$ is the size of the key space
of a single encryption.
An interesting corollary is that the most time-efficient attack is not known to give the best gain
for time-space product. An attacker can pay with more time and save on space, but with a ratio
that differs fundamentally from the classical case.

We have then studied the case of $4$-encryption, for which we have given a quantized version
of the dissection attack of Dinur, Dunkleman, Keller and Shamir. This quantization uses the framework
of quantum walks, another important concept of quantum algorithms and complexity theory. 
We do not not know if this quantum attack is optimal with respect to the number of queries or time,
and we have only compared it with the best currently know classical attack.
But we can already conclude that this quantization gives better gains in time and time-space product
than the optimal algorithm for 2-encryption.

This seems to indicate that successive encryptions is less resistant against quantum attacks 
than it is
against classical ones. This intuition could be confirmed by proving lower bounds on quantum attacks,
eventually allowing to consider the most time-efficient attacks in both cases.
We have not been able to answer this question for 4-encryption as we did for 2-encryption.
Of course, one possibility is that the dissection algorithm is not optimal,
but improving it is a challenging question in classical cryptography. 
Proving a quantum lower bound or improving the classical upper bounds would
both give more insight on security amplification.

Another crucial element that misses in our analysis is a study of the $r$-encryption case.
Classically, the dissection algorithm that we have discussed can be generalized
iteratively applied, leading to attacks on $r$-encryption for arbitrary values of $r$.
In our case, there is a barrier that prevents the scaling of our results.
The composition theorem derived from the generalized adversary method, or the
quantum-walk framework both introduce multiplicative factors, and
it seems harder to understand how the complexity grows as a function of $r$.

Our work can be understood as a proof of principle that quantum complexity and algorithms theory
has developed powerful tools to tackle important questions in cryptography and post-quantum cryptography.
Maybe a better use of these tools can lead to improving and extending the results presented here. Or maybe these
tools need to be sharpened in order to apply to specific cryptographic situations.
In both cases, we hope that our work will motivate further interactions between classical cryptographers
and quantum computer scientists. Such interactions seem crucial to establish a serious approach to
post-quantum cryptography.
In particular, iterating block ciphers is not a very good procedure to amplify security against classical
adversary. We hope that the techniques presented here will be applied to 
more efficient cryptographic procedures in the future.

\section*{Aknowledgements}

The author is grateful to Anne Canteaut, Maria Naya-Plasencia, Fr\'ed\'eric Magniez, Robin Kothari, Andris Ambainis and Gilles Brassard for stimulating discussions and pointers to references.
M.K. acknowledges financial support from acknowledge financial support from ANR retour des post-doctorants NLQCC (ANR-12-PDOC-0022- 01).

\bibliographystyle{alpha}
\bibliography{biblio}

\appendix

\section{Quantum query complexity}

\subsection*{The model}
In the quantum query complexity model, the goal is to compute some function $F:S \rightarrow T$ on input $x$.
For simplicity, we assume in this short presentation that
$S\subseteq \{0,1\}^n$.
The input is given as a black-box and accessed through {\em queries}.
A classical query $\ell$ to the input $x$ returns $x_\ell$.

In the quantum setting, the algorithm is executed
on an architecture with three quantum registers: an input register $I$, a query register $Q$ and a working register $W$.
The state of the quantum computer when the algorithm starts is ${\ket x}_I{\ket{0,0}}_Q{\ket 0}_W$.
There are several equivalent formalism to model quantum queries.
Without loss of generality,
we consider a quantum query as an operator $$O:{\ket x}_I \ket{i,y}_Q  \longmapsto \ket {x}_I \ket{i, y + x_i \mod 2}_Q.$$
Considering non-boolean inputs is equivalent, up to logarithmic factors.

The quantum algorithm supposed to compute $f$ alternates between quantum query operations $O_i$ and work operations that are
unitaries $U_i$ acting on the query register and the working register.
After $t$ queries to the input, the state of the quantum computer
is $U_t O_t U_{t-1} \ldots O_2 U_1 O_1 {\ket x}_I{\ket{0,0}}_Q{\ket 0}_W$.

We assume that the working register contains $\log |T|$ qubits to encode the value $f(x)$.
The last step of the algorithm is then to measure these qubits and output the value that is obtained.
Finally, an immediate corollary of the model
is that the time complexity of $f$ is at least equal to the query complexity of the algorithm.
The generalized adversary method can be used to prove a lower bound on query complexity, which in turn
bounds the time complexity.

\subsection*{Quantum walks}

We first review the paradigm of quantum walks.
We use this tool to design an attack against 4-encryption. More precisely,
we use it to combine an exhaustive search with a collision finding algorithm.
The framework of quantum walks allows to easily analyze the ressources used
by the algorithm.

A search algorithm aims to find a marked element in a finite set $U$.
Classically, it
can be seen as a walk on a graph whose nodes are indexed by subsets of elements of $U$.
Each step of the algorithm consists in 
walking from one node to another. The algorithm can also maintain a data structure that
is updated at each step. After each move, the algorithm checks if the node contains
a marked element. The algorithm terminates when a marked element is found.

Magniez, Nayak, Roland and Santha have designed a generic theorem in order to quantize
search algorithms expressed as random walks.
The cost of the resulting quantum algorithm can be written as a function of 
{\SSS}, {\UU} and {\CC}. These are the cost of
setting up the quantum register in a state that corresponds to the stationary distribution,
moving unitarily from one node to an adjacent node,
and checking if a node is marked,
respectively. 

\begin{theorem}~\cite{MNRS}
\label{thm:MNRS}
Let $P$ be an ergodic and reversible Markov chain with eignvalue gap $\delta$, and let $\eps>0$
be a lower bound on the probability that an element chosen from the stationary distribution
of $P$ is marked, whenever the set of marked element is non-empty. Then, there exists a quantum algorithm which
finds, with high probability, a marked element if there is any at cost of order $\SSS+\frac 1 {\sqrt \eps}(\frac 1 {\sqrt \delta} \UU +\CC)$.
\end{theorem}

Grover's algorithm can be seen as a trivial application of Theorem~\ref{thm:MNRS} (see also~\cite{santha}).
The underlying graph is the complete graph whose nodes are indexed by elements of $U$, with no data structure.
In our case however, the checking procedure is very different because it implies multiple queries to the input, and
we use Theorem~\ref{thm:MNRS} to design our attack.

\subsection*{The generalized adversary method}

We use the generalized adversary method to prove a lower bound on the problem $KE_2$.
The intuition of the method is to consider pairs of inputs leading to different outputs. Each
pair is given a weight\footnote{The original method uses probability distribution, the generalized method
allows for {\it negative weights} as well.}
 quantifying how difficult it is to distinguish them. The key point of the method
is then to measure the progress made by a single quantum query to distinguish paris of inputs.
This intuition can be formalized in several different ways. We use the spectral version,
an elegant algebraic formalization of the previous intuition~\cite{hls07}.

\begin{definition}
\label{def:adv}
Fix a function $F:S\rightarrow T$, with $S \subseteq \{0,1\}^n$.  
A symmetric matrix $\Gamma: S\times S \rightarrow \mathbb{R}$
is an adversary matrix for $F$ provided $\Gamma[x,y]=0$ whenever $F(x)=F(y)$.
Let $\Delta_\ell[x,y] = 1 $ if $x_\ell \neq y_\ell$ and 0 otherwise.
The adversary bound of $F$ using $\Gamma$ is $$\advpm(F;\Gamma) = \min_\ell \frac{\| \Gamma\|}{\|\Gamma \bullet \Delta_\ell\|},$$ where $\bullet$ denotes entrywise (or Hadamard) product, and
$\|A\|$ denotes the spectral norm of $A$.
The adversary bound $\advpm(F)$ is the maximum, over all adversary matrices $\Gamma$ for $F$,
of $\advpm(F;\Gamma). $
\end{definition}

H\o yer, Lee and \v Spalek introduced the generalized adversary method and proved
that the adversary value is a lower bound on quantum query complexity~\cite{hls07}. In our proof,
we need the fact that the adversary bound characterizes quantum query complexity, up to a constant factor.
Our proof starts by considering an adversary matrix for the problem $CF$ such that the adversary bound
is equal to the quantum query complexity.
The fact that the generalized adversary method is tight was proven by
Lee, Mittal, Reichardt, \v Spalek and Szegedy.

\begin{theorem}~\cite{LMRSS10}
\label{thm:advq}
Fix a function $F:S\rightarrow T$.  The bounded-error quantum query complexity of $F$ is characterized by the general adversary bound:
$$Q(f) = \Theta(\advpm(F)).$$
\end{theorem}

\end{document}